\documentclass[letterpaper,11pt]{article}

\usepackage[margin=1in]{geometry}
\usepackage{amsmath, amsfonts, amssymb, amsthm}
\usepackage{bbm,url}
\usepackage[colorlinks=true,urlcolor=black,linkcolor=black,citecolor=black]{hyperref}
\usepackage{graphicx}
\usepackage{color}
\usepackage{subcaption}
\usepackage{hhline}
\usepackage{pifont}
\usepackage{nicefrac,bm}
\usepackage{mathtools}
\usepackage{algorithm}
\usepackage[noend]{algpseudocode}
\usepackage{thmtools}
\usepackage{thm-restate}

\def\stack#1#2{{\substack{#1 \\ #2}}}
\def\midd{\:\middle|\:}
\def\set#1{\left\{ #1 \right\}}

\def\prn#1{\left( #1 \right)}
\def\argmax{\operatornamewithlimits{arg\,max}}

\allowdisplaybreaks

\newtheorem{theorem}{Theorem}
\newtheorem*{theorem*}{Theorem}

\newtheorem{lemma}{Lemma}

\theoremstyle{definition}
\newtheorem{definition}{Definition}

\theoremstyle{condition}

\newcommand{\hide}[1]{}

\newcommand{\s}[1]{\mathsf{#1}}

\newcommand{\x}{\mathbf{x}}
\newcommand{\y}{\mathbf{y}}

\newcommand{\N}{\mathbb N}
\newcommand{\Z}{\mathbb Z}
\newcommand{\Zp}{\mathbb{Z}^+}

\newcommand{\R}{\mathbb R}

\title{(Almost) Envy-Free, Proportional and Efficient Allocations of an Indivisible Mixed Manna\footnote{Supported by NSF grant CCF-1750436 (CAREER)}}

\author{Vasilis Livanos\footnote{University of Illinois at Urbana-Champaign, USA} \\
\texttt{\small livanos3@illinois.edu} 
\and
Ruta Mehta\footnote{University of Illinois at Urbana-Champaign, USA}\\
\texttt{\small rutameht@illinois.edu}
\and
Aniket Murhekar\footnote{University of Illinois at Urbana-Champaign, USA}\\
\texttt{\small aniket2@illinois.edu}
}

\date{}
\begin{document}

\maketitle

\begin{abstract}
We study the problem of finding fair and efficient allocations of a set of indivisible items to a set of agents, where each item may be a good (positively valued) for some agents and a bad (negatively valued) for others, i.e., a mixed manna. As fairness notions, we consider arguably the strongest possible relaxations of envy-freeness and proportionality, namely  envy-free up to any item (EFX and EFX$_0$), and proportional up to the maximin good or any bad (PropMX and PropMX$_0$). Our efficiency notion is Pareto-optimality (PO).

We study two types of instances: 
$(i)$ {\em Separable}, where the item set can be partitioned into goods and bads, and $(ii)$ {\em Restricted mixed goods (RMG)}, where for each item $j$, every agent has either a non-positive value for $j$, or values $j$ at the same $v_j>0$. We obtain polynomial-time algorithms for the following:
 \begin{itemize}
     \item Separable instances: PropMX$_0$ allocation.
    \item RMG instances: Let  {\em pure bads} be the set of items that everyone values negatively.
    \begin{itemize}
        \item 
    PropMX allocation for general pure bads.
    \item EFX+PropMX allocation for  identically-ordered pure bads.
    \item  EFX+PropMX+PO allocation for identical pure bads. 
    \end{itemize}
 \end{itemize}

Finally, if the RMG instances are further restricted to {\em binary mixed goods} where all the $v_j$'s are the same, we strengthen the results to guarantee EFX$_0$ and PropMX$_0$ respectively.

\end{abstract}

\newcommand{\efx}{EFX}
\newcommand{\efxzero}{$\text{EFX}_0$}
\newcommand{\propm}{PropM}
\newcommand{\propx}{PropX}
\newcommand{\propmzero}{$\text{PropM}_0$}
\newcommand{\prop}{\s{Prop}}
\newcommand{\propmx}{PropMX}
\newcommand{\propmxzero}{$\text{PropMX}_0$}

\let\oldnl\nl
\newcommand{\nonl}{\renewcommand{\nl}{\let\nl\oldnl}}

\newtheorem{innercustomthm}{Theorem}
\newenvironment{customthm}[1]
  {\renewcommand\theinnercustomthm{#1}\innercustomthm}
  {\endinnercustomthm}


\section{Introduction}\label{sec:intro}

The problem of fair division is concerned with allocating items to agents in a \textit{fair} and \textit{efficient} manner. Formally 
introduced by Steinhaus \cite{steinhaus}, fair division is an active area of research studied across fields like computer science and 
economics. Most work has focused on the fair division of \textit{goods}: items which provide non-negative \textit{value} (or utility) to the agents to whom they are allocated. However, several practical scenarios involve \textit{bads} (or chores). Bads are items which impose a \textit{cost} (or disutility) to the agent to whom they are allocated. Generalizing both settings, we study fair division of a set $M$ of indivisible items, where each $j\in M$ can be a good for some agents and a bad for others -- a {\em mixed manna}. Examples of mixed manna include splitting assets and liabilities when dissolving a partnership, dividing tasks among various team members, and deciding teaching assignments between faculty. 
The valuation of agent $i$ for a set $S\subseteq M$ of items is defined by an additive function $v_i(S)=\sum_{j\in S}v_{ij}$ where $v_{ij}\in \mathbb R$ is the value agent $i$ has for item $j$. We say that item $j$ is a good for agent $i$ if $v_{ij}\ge 0$, otherwise it is a bad for her.

Arguably, the two most popular fairness notions are of  envy-freeness (EF) \cite{foleyEF, VARIAN1974-efpo} and proportionality (Prop) \cite{steinhaus}. EF requires that every agent 
(weakly) prefers her own allocation than anyone else's, while Prop requires that every agent gets her proportional value, i.e.,$\frac{1}{n}$-fraction of her value for all the items. When items are \textit{divisible}, both EF and Prop allocations are known to exist. However, in the case of \textit{indivisible} items, neither may exist for very simple examples, like allocating one good among two agents who value it equally. 

One of the strongest relaxations of EF for indivisible items is
{\em EF up to any item (EFX)} \cite{caragiannis16nsw-ef1}. That is, for goods manna, every agent $i$ does not envy any other agent $k$ after removal of a (positively-valued) good from $k$'s bundle, and for bads manna the envy vanishes after removal of any bad from $i$'s bundle. Generalizing these to mixed manna, we say that an allocation is EFX if each agent $i$ does not envy another agent $k$ after either removal of a positively-valued good from $k$'s bundle or removal of a bad from $i$'s bundle. 

The analogous relaxation for proportionality is the notion of \textit{proportionality up to any item} (PropX), defined by Aziz
et al. \cite{AzizMoulinSandom}. For goods, an allocation is said to be PropX if every agent can receive her proportional share
after the addition to her bundle of any one good not in her bundle. Similarly, for bads, an allocation is said to be PropX if every 
agent can receive her proportional share after removing any one bad assigned to her. While PropX need not exist for goods 
\cite{MoulinAnnual, AzizMoulinSandom}, Li et al. \cite{LiLiWu} showed that a PropX allocation always exists for bads and can be
computed in polynomial time. Since PropX need not exist for goods, Baklanov et al. \cite{BaklanovGGS} introduce the notion of
\emph{proportionality up to the maximin good (PropM)}, a weaker version of PropX.  
They show that PropM allocations always exist
for the setting of goods and agents with additive valuations, and that they can be computed in polynomial time.

Relaxations of Prop have not been studied previously for indivisible mixed manna. Combining the strongest possible guarantees, we define the following relaxation for mixed manna: 
\emph{proportionality up to the maximin good or any bad (PropMX)}. Informally, an allocation is said to be PropMX if every agent can receive her proportional share after the addition to her bundle of the \emph{maximin} positively-valued good for that agent,
or after the removal of any one bad assigned to her.

In addition to being fair, it is important for allocations to be \textit{efficient}. The standard notion of economic efficiency is Pareto optimality (PO). An allocation is said to be PO if no other allocation makes an agent better off without making someone else worse off.

In this paper we ask if it possible to achieve all of EFX, PropMX, and PO in a single allocation. And if yes, then can it be computed in polynomial time. The answer is not known even for goods (bads) only manna. In particular, existence of an EFX allocation is a celebrated open question even for goods (bads) manna \cite{procaccia20}. And therefore the focus has been on resolving special cases (see Section~\ref{sec:related} for a detailed discussion). 
Furthermore, it has been noted that mixed manna is significantly harder to handle than the goods (bads) manna \cite{bogomolnaia2017manna, KulkarniMehtaTaki}. Next we describe our contributions in this context.

\subsection{Our Contributions.}

We study the problem of computing \efx+PO and \propmx+PO allocations for mixed manna instances. Given that computing EFX (and EFX+PO) allocations is a challenging problem even for general instances of goods, a significant amount of past works have focused on sub-classes by restricting the values that agents have for the items, see~\cite{amanatidis2020mnwefx,murhekar2021aaai,barman2018binarynsw,garg2021sagt, aleksandrov2019greedy,aleksandrov2020algorithms} and references therein for practical scenarios involving different sub-classes.
Among these, the valuation classes of identical \cite{plaut2018efx,barman2018binarynsw,aleksandrov2019greedy,aleksandrov2020algorithms}, identical order preferences (IDO) \cite{plaut2018efx,LiLiWu}, restricted additive \cite{bansal2006santaclaus}, and binary \cite{barman2018binarynsw,nisarg20binaryonerule,bhaskar2020chores} are well-studied for goods (bads) manna. We study all of these for mixed manna. 

First, to extend these to mixed manna we need to partition the items into three sets: the set $M^+$ of \textit{mixed goods}, which are valued positively by at least one agent; the set $M^0$ of \textit{dummy bads}, which are not valued positively by any agent but may be valued at zero by some; and set $M^-$ of \textit{pure bads} which are valued negatively by all agents. 

We consider the following:
\begin{itemize}
    \item \emph{Separable} instances: For every item $j$, all the agents value it either non-negatively or negatively. That is the item set can be partitioned in to goods (non-negatively valued) and bads (negatively valued).  
    \item \emph{Restricted Mixed goods} instances: For every item $j \in M^+$ there exists a value $v_j > 0$ such that if an agent values $j$ positively, then she values it at $v_j$. 
    Furthermore, if $v_j = v_{j'}$ for all $j, j' \in M^+$, then the instance is called a \emph{binary mixed goods} instance.
\end{itemize}

Note that the separable and RMG instances are incomparable. 
We obtain following results for these settings.
\begin{itemize}
    \item For \textbf{Restricted Mixed Goods}, we give polynomial-time algorithms for computing allocation that is
    \begin{itemize}
        \item \propmx+\efx+PO when agents have identical valuation over the pure bads $M^-$ (Theorem~\ref{thm:efxpo}).
        \item \propmx+\efx~when agents have identical ordinal preference (IDO) over the pure bads (Theorem~\ref{thm:restrictedgoodsefx-ido}).
        \item \propmx~for the case of \textbf{general pure bads} (Theorem~\ref{thm:restrictedgoodsefx-gen}).
    \end{itemize}
    \item For the special case of \textbf{binary mixed goods}, we strengthen the previous results to \efxzero~and \propmxzero~ (Theorem~\ref{thm:efxzeropo}) respectively.\footnote{\efxzero~is a stronger notion where the EFX condition wrt goods also considers zero valued items. Similarly, \propmxzero~is a stronger notion where the maximin good in the definition of PropMX is allowed to have zero value.}
    \item In showing the above results, we show that an EFX+PO allocation can be computed in polynomial-time for a \textbf{goods} only manna with \textbf{restricted valuations} (Theorem~\ref{thm:restrictedgoodsefx}). This class is orthogonal to the class of IDO goods, for which EFX allocations are known to be efficiently computable~\cite{plaut2018efx}. 
    \item For \textbf{separable} instances, we present an algorithm which returns, in polynomial time, an allocation that is \propmxzero~(Theorem~\ref{thm:ido-separ}).
\end{itemize}

We note that as a corollary we obtain some of the results of ~\cite{aleksandrov2019greedy} and \cite{aleksandrov2020algorithms} for the binary and identical mixed manna. We observe that \efx~(resp. \efxzero) implies \propmx~(resp. \propmxzero), and therefore whenever we get \efx~we get \propmx, but not vice-versa. Furthermore, we show via a counterexample that one cannot hope to obtain a \propmxzero+PO allocation, even for the goods manna (Appendix~\ref{app:counter}), and hence the first two results for the RMG instances can not be strengthened.

The subclasses we consider are interesting from a practical viewpoint as well. In many settings, there is a subset of “interested agents” for an item who all value the item the same, and the rest of the agents do not value it. For example, in a partnership dissolution, a risky asset $j$ could be seen as a good by some agents, who all value it at $v_j$ (because the forecasted return is $v_j$), but other agents might not be interested or even value it negatively due to the inherent risk. This is captured by restricted (mixed) goods. Likewise, identical and binary preferences often arise in practice; binary valuations can be used by agents to indicate approval or indifference towards a particular good.

\subsection{Other Related Work}\label{sec:related}

\paragraph{\bf EF1 and EFX.} The envy-cycle algorithm~\cite{lipton} showed that allocations that are envy-free up to the removal of the most valued good (EF1) can be computed in polynomial-time; EF1 is a weaker fairness notion than EFX. This algorithm was adapted to show that EFX allocations for IDO goods can be computed in polynomial-time~\cite{plaut2018efx}. For bads, ~\cite{bhaskar2020chores} showed a variant of envy-cycle algorithm (using the top-envy graph) can be adapted for computing an EF1 allocation for separable instances (which they refer to as doubly monotone instances). In the goods setting, EFX allocations exist for 3 agents~\cite{chaudhury2020efx}, and for the class of bivalued instances~\cite{amanatidis2020mnwefx}.

\paragraph{\bf EF1/EFX + PO.} 
The following results are for the case of goods. Barman et al.~\cite{Barman18FFEA} show that EF1+PO allocations exist for and can be computed in pseudo-polynomial time. Recently, \cite{murhekar2021aaai} showed that for $k$-ary instances with constant $k$ or for instances with constantly many agents, an EF1+PO allocation can be computed in polynomial-time. For bivalued instances, an EFX+PO allocation is poly-time computable~\cite{garg2021sagt}. For bivalued instances of chores,~\cite{garg22chores} showed that an EF1+PO is polynomial-time computable. For a mixed manna with only two agents,~\cite{aziz2019mixedmanna} showed that an EF1+PO allocation is polynomial-time time computable.

\paragraph{\bf Mixed Manna.}
For mixed manna, \cite{AzizMoulinSandom} showed that a Prop1+PO allocation of an indivisible mixed manna is polynomial-time computable. For the special cases of identical or ternary utilities, \cite{aleksandrov2019greedy,aleksandrov2020algorithms} showed that an EFX+PO allocation can be polynomial-time. The competitive division of a mixed manna has also been studied~\cite{garg2020mixedmanna,bogomolnaia2017manna}.

\subsection{Organization}

Section~\ref{sec:prelim} sets up various technical preliminaries on fairness notions and their relations.
Section~\ref{sec:efxpo} describes our results for the RMG setting with identical bads. In Section~\ref{sec:separ}
we discuss our result for Separable instances. In Section~\ref{sec:efx-gen} we extend our results for the RMG setting for IDO and general
bads. We discuss the results along with some open questions in Section~\ref{sec:discussion}. Appendix~\ref{app:counterexamples}
contains two counterexamples of interest.

\section{Preliminaries}\label{sec:prelim}

\paragraph{\bf Problem instance.} 
A \emph{fair division instance} is a tuple $(N,M,V)$, where $N = [n]$ is a set of $n\in\N$ agents,
$M = [m]$ is a set of $m\in\N$ indivisible items, and $V = \{v_1,\dots,v_n\}$ is a set of utility
functions, one for each agent $i\in N$. Each utility function $v_i : M \rightarrow \mathbb{R}$ is
specified by $m$ numbers $v_{ij} \in \R$, one for each item $j\in M$, which denotes the value agent $i$
has for receiving item $j$. When $v_{ij} \geq 0$ for every $i \in N, j \in M$, we call the instance a
\emph{goods} instance. In the case of $v_{ij} < 0$ for every $i \in N, j \in M$, we call the instance a
\emph{bads} instance. Finally, when we place no restrictions on $v_{ij}$, the instance is called a
\emph{mixed manna} instance. We assume that the value functions are additive, that is, for every agent
$i \in N$, and for $S \subseteq M$, $v_i(S) = \sum_{j\in S} v_{ij}$. For notational ease, we write
$v(S - j)$ instead of $v(S\setminus \{j\})$ and $v(S + j)$ instead of $v(S\cup \{j\})$. Throughout this paper, unless stated otherwise, we assume every fair division instance is a mixed manna instance.

\paragraph{\bf Partitioning the mixed manna.} In the mixed manna setting, a useful way of partitioning the set of items $M$ is into the three sets $M^+, M^0$, and $M^-$, where:

\begin{enumerate}
\item $M^+ = \{j \in M : \exists i \in N, \: \: v_{ij} > 0\}$ is the set of \textit{mixed goods},
\item $M^0 = \{j \in M : \forall i \in N, \: \: v_{ij} \leq 0 \text{ and } \exists i \in N, \: \: v_{ij} = 0\}$ is the set of \textit{dummy bads}, and
\item $M^- = \{j \in M : \forall i \in N, \: \: v_{ij} < 0\}$ is the set of \textit{pure bads}.
\end{enumerate}

In other words, $M^+$ comprises of the items which are pure goods for some agent, and may be bads or dummies for others; $M^0$ comprises of items which are dummy for some agent, and are not goods for anyone; $M^-$ comprises of items which are bads for everyone.

\paragraph{\bf Instance types.}
We call a fair division instance $(N, M, V)$ a \emph{separable instance} if we can partition $M$ into $M^{\geq 0}$ and $M^-$, where $M^{\geq 0} := \{j \in M \mid \forall i \in \allowbreak N, \: \: v_{ij} \geq 0\}$, the set of items which are not bads for any agent, and $M^-$ is the set of pure bads.

We also define the setting of \emph{restricted mixed goods}, in which for every $j \in M^+$,
there exists a value $v_j > 0$ such that for all $i \in N$, if $v_{ij} > 0$, then $v_{ij} = v_j$
(notice that if $v_{ij} \leq 0$ for a mixed good $j$, then no restrictions are placed on $v_{ij}$). A
special case of the restricted mixed goods setting is the \emph{binary mixed goods} setting, where
for all $j, j' \in M^+$, $v_j = v_{j'}$.

An instance is called \emph{identical ordering (IDO)}, if all agents have the same ordinal preference for all items, i.e., there exists an ordering of the items in $M$ such that for all agents $i \in N$, $v_{i1} \leq v_{i2} \leq \dots \leq v_{im}$. A special case of an IDO instance is the \emph{identical} setting, in which for every $j \in M$, $v_{ij} = v_{i'j}$ for all
$i,i' \in N$. 

\paragraph{\bf Allocation.} An \emph{allocation} $\x$ of items to agents is an $n$-partition $\x_1, \dots, \x_n$ of the items, where agent $i$ is allotted the bundle $\x_i \subseteq M$, and gets a total utility of $v_i(\x_i)$.

\paragraph{\bf Pareto-optimality.} An allocation $\y$ Pareto-dominates an allocation $\x$ if $v_i(\y_i) \geq v_i(\x_i), \forall i \in N$
and there exists $h\in N$ s.t. $v_h(\y_h) > v_h(\x_h)$. An allocation is said to be \textit{Pareto-optimal} (PO) if no Pareto-allocation dominates it.

\paragraph{\bf Welfare functions.} Given an allocation $\x$:
\begin{enumerate}
\item the (utilitarian) social welfare $\s{SW}(\x)$ of $\x$ is the sum of agents' utilities under $\x$, i.e., $\s{SW}(\x) = \sum_{i\in N} v_i(\x_i)$.
\item the Nash welfare $\s{NW}(\x)$ of $\x$ is the geometric mean of the agents' utilities under $\x$, i.e., $\s{NW}(\x) = (\prod_i v_i(\x_i) )^{1/n}$.
\end{enumerate}
Any allocation $\x$ that maximizes the social welfare or the Nash welfare is Pareto-optimal, since a dominating allocation will have higher welfare than $\x$, which is not possible. 

\paragraph{\bf Fairness notions.} We now define the fairness notions of interest.

\begin{definition}\label{def:efx}(Envy-freeness and its relaxations.)
An allocation $\x$ is said to be:
\begin{enumerate}
\item \textit{Envy-free} if for all $i, h \in N$, $v_i(\x_i) \geq v_i(\x_h)$. 
\item \emph{Envy-free up to
any item} (EFX) if for all $i,h \in N$ either 
\begin{enumerate}
\item[(i)] $v_i(\x_i) \geq v_i(\x_h - g) \quad \forall g \in \x_h \text{ s.t. } v_{ig} > 0$, or
\item[(ii)] $v_i(\x_i - c) \geq v_i(\x_h) \quad \forall c \in \x_i \text{ s.t. } v_{ic} < 0$.
\end{enumerate}
\item $\text{EFX}_0$ if for all $i,h \in N$ either 
\begin{enumerate}
\item[(i)] $v_i(\x_i) \geq v_i(\x_h - g) \quad \forall g \in \x_h \text{ s.t. } v_{ig} \ge 0$, or
\item[(ii)] $v_i(\x_i - c) \geq v_i(\x_h) \quad \forall c \in \x_i \text{ s.t. } v_{ic} < 0$.
\end{enumerate}
\end{enumerate}
The difference between the definitions of EFX and \efxzero~is that EFX allows for the envy of an agent $i$ towards agent $h$ to disappear after removing any positively-valued item from the bundle of $h$, whereas in \efxzero~this envy must disappear after removing any non-negative valued item. Thus, \efxzero~is a stronger notion than \efx, and it is easy to see that any \efxzero~allocation is \efx, but not vice-versa.

We say that an agent $i$ \textit{envies} an agent $h$ if $v_i(\x_i) < v_i(\x_h)$. Likewise, we say that an agent $i$ \emph{\efx-envies} (resp. \efxzero-envies) an agent $h$ if neither conditions (i) nor (ii) of (2) (resp. (3)) hold for $i$ with respect to $h$.
\end{definition}

\begin{definition}\label{def:envy-graph}(Envy Graph.) The \textit{envy-graph} of an allocation $\x$ is a directed graph
$G_\x = (N,E)$ where each agent is a node and there exists an edge from agent $i$ to agent $h$ if and only if
$v_i(\x_i) < v_i(\x_h)$, i.e. if and only if $i$ envies $h$. The \textit{top envy-graph} $G^*_\x = (N,E)$ is a directed graph where 
there is an edge from $i\in N$ to $h\in N$ if $i$ envies $h$ and additionally $h \in \s{argmax}_{i' \neq i} {v_i(\x_{i'})}$.

A \textit{source} in $G_\x$ is an agent with in-degree zero, i.e., an agent who nobody envies. A \textit{sink} in $G_\x$ is an agent with 
out-degree zero, i.e., and agent who envies nobody. 
Notice that if there exists a (directed) cycle $C$ in the envy-graph for some allocation, then we can reallocate bundles among the agents in $C$ in the reverse order and all agents in $C$ receive a bundle that they prefer
over their current bundle, while the utility of agents outside of $C$ does not change. This procedure is called \emph{envy-cycle
elimination} and results in a Pareto-improvement.
\end{definition}

\begin{definition}\label{def:propm}(Proportionality and its relaxations.) An allocation $\x$ is said to be:
\begin{enumerate}
\item \textit{Proportional} if for all agents $i \in N$, we have $v_i(\x_i) \geq \prop_i$, where $\prop_i = \frac{1}{n} \cdot v_i(M)$ is the \textit{proportional share} of an agent. 
\item \emph{Proportional up to the maximin
good} (\propm) for a goods instance if, for all $i \in N$:
\[
v_i(\x_i) + d_i(\x) \geq \prop_i,
\]
where $d_i(\x) = \max_{i' \neq i} \min_\stack{j \in \x_{i'}}{v_{ij} > 0} {v_{ij}}$. The corresponding item maximizing the expression of $d_i(\x)$ is called the \textit{maximin good} of $i$ for $\x$.
\item \propmzero~for a goods instance if for all $i \in N$:
\[
v_i(\x_i) + d_i(\x) \geq \prop_i,
\]
where $d_i(\x) = \max_{i' \neq i} \min_\stack{j \in \x_{i'}}{v_{ij} \ge 0} {v_{ij}}$.
Note that \propmzero~is a more demanding condition than \propm, and it is easy to see that any \propmzero~allocation is \propm, but not vice versa.

\item \emph{Proportional up to any bad} (\propx) for a bads instance if for all $i \in N$ and $\forall c \in \x_i$:
\[
v_i(\x_i - c) \geq \prop_i.
\]
One can analogously define \propx~for goods, however PropX allocations of goods need not always exist. We include an example due to
Aziz, Moulin and Sandomirskiy \cite{MoulinAnnual, AzizMoulinSandom} in Appendix~\ref{app:counter-propx}.
For the mixed manna setting, we combine the definitions above:
\item \textit{Proportional up to the maximin good or any bad} (PropMX) for a MiXed manna instance, if for all $i\in N$ either:
\begin{enumerate}
\item[(i)] $v_i(\x_i) + d_i(\x) \geq \prop_i$, where $d_i(\x) = \max\limits_{i' \neq i} \min\limits_\stack{j \in \x_{i'}}{v_{ij} > 0} {v_{ij}}$, or
\item[(ii)] $\forall c \in \x_i$ such
that $v_{ic} < 0$, $v_i(\x_i - c) \geq \prop_i$.
\end{enumerate}
\item \propmxzero~for a MiXed manna instance, if for all $i\in N$ either:
\begin{enumerate}
\item[(i)] $v_i(\x_i) + d_i(\x) \geq \prop_i$, where $d_i(\x) = \max\limits_{i' \neq i} \min\limits_\stack{j \in \x_{i'}}{v_{ij} \ge 0} {v_{ij}}$, or
\item[(ii)] $\forall c \in \x_i$ such
that $v_{ic} < 0$, $v_i(\x_i - c) \geq \prop_i$.
As before, any \propmxzero~allocation is \propmx, but not vice-versa.
\end{enumerate} 
\end{enumerate}
\end{definition}

\paragraph{\bf Relating the fairness properties.} First notice that for any mixed manna instance, every agent $i$ that has
$v_i(M) \leq 0$ can be trivially \propmxzero~satisfied by allocating no items to them. Therefore, for the remainder of the paper,
we assume that $v_i(M) > 0$ for all agents.

It is easy to see that in the case of an instance with bads only, \efx~(resp.\efxzero) implies \propx (\cite{LiLiWu}, Lemma 3.2). We 
extend this observation and show that this implication continues to hold in the mixed manna setting.

\begin{lemma}\label{lem:efx-propm}
Consider a mixed manna instance and let $\x$ be an \efx~allocation for that instance. Then $\x$ is also \propmx. Further, if $\x$ is \efxzero, then $\x$ is also \propmxzero.
\end{lemma}
\begin{proof}
Since $\x$ is \efx, we know that for all $i,h \in N$, either
\begin{equation}\label{eq:efx-pos}
v_i(\x_i) \geq v_i(\x_h) - v_{ig} \quad \forall g \in \x_h \text{ s.t. } v_{ig} > 0, 
\end{equation}
or
\begin{equation}\label{eq:efx-neg}
v_i(\x_i) - v_{ic} \geq v_i(\x_h) \quad \forall c \in \x_i \text{ s.t. } v_{ic} < 0.
\end{equation}
Fix an agent $i \in N$ and let $N^+ \subseteq N$ denote the set of agents that $i$ does not envy up to any item because of \eqref{eq:efx-pos}, and $N^- \subseteq N$ denote the set of agents that $i$ does not envy up to any item because of \eqref{eq:efx-neg}. If, in \eqref{eq:efx-pos}, we select $g \in \x_h$ to be the minimum-value pure good of
$\x_h$ according to $i$'s valuation function, and then select the maximum such good over all $h \in N^+$, we get that this item is the maximin good for agent $i$, and thus, for all $h \in N^+$
\begin{equation}\label{eq:efx-pos2}
v_i(\x_i) + d_i(\x) \geq v_i(\x_h).
\end{equation}
We sum up \eqref{eq:efx-pos2} for all $h \in N^+$ and \eqref{eq:efx-neg} for all $h \in N^-$, making sure the case 
of $h = i$ is counted only once, and then sum both inequalities together, and we get that, for all
$c \in \x_i$ such that $v_{ic} < 0$,
\begin{align*}
n \cdot v_i(\x_i) + |N^+| \cdot d_i(\x) - |N^-| \cdot v_{ic} &\geq v_i(M) \iff \\
v_i(\x_i) + \frac{|N^+|}{n} \cdot d_i(\x) - \frac{|N^-|}{n} \cdot v_{ic} &\geq \frac{1}{n} \cdot v_i(M) = \prop_i.
\end{align*}
Notice that, for any $c \in \x_i$ such that $v_{ic} < 0$, $\frac{|N^+|}{n} \cdot d_i(\x) - \frac{|N^-|}{n} \cdot v_{ic}$ is a convex
combination of $d_i(\x)$ and $-v_{ic}$, and thus $\frac{|N^+|}{n} \cdot d_i(\x) - \frac{|N^-|}{n} \cdot v_{ic} \leq \max\{d_i(\x), -v_{ic}\}$.
Therefore, either
\[
v_i(\x_i) + d_i(\x) \geq \prop_i,
\]
or, for all $c \in \x_i$ such that $v_{ic} < 0$,
\[
v_i(\x_i - c) \geq \prop_i,
\]
thus showing that $\x$ is \propmx.

Finally, if $\x$ is \efxzero, all inequalities still hold if we let
$d_i(\x) = \max_{i' \neq i} \min_\stack{j \in \x_{i'}}{v_{ij} \geq 0} {v_{ij}}$, and thus it follows that $\x$ is also \propmxzero.
\end{proof}

The following lemma shows allocating items to agents who value them at the highest possible value results in a PO allocation. For any $j\in M$, let $\eta_j = \max_{i\in N} v_{ij}$.
\begin{lemma}\label{lem:po} Let $\x$ be an allocation in which every item $j\in M$ is allocated to an agent $i$ s.t. $v_{ij} = \eta_j$. Then $\x$ is PO and maximizes the social welfare.
\end{lemma}
\begin{proof}
The social welfare of $\x$ is $\s{SW}(\x) = \sum_{i\in N} \sum_{j\in \x_i} v_{ij} = \sum_{i\in N} \sum_{j\in \x_i} \eta_j = \sum_{j\in M} \eta_j$. Further, for any allocation $\y$, $\s{SW}(\y) = \sum_{i\in N} \sum_{j\in \y_i} v_{ij} \le \sum_{i\in N} \sum_{j\in \y_i} \eta_j \le \sum_{j\in M} \eta_j$, since $v_{ij}\le \eta_j$ for every $i,j$. Hence, $\x$ maximizes the social welfare.

If $\x$ is not PO, then there must be an allocation $\y$ which dominates $\x$, i.e., $v_i(\y_i)\ge v_i(\x_i)$ for every $i\in N$ and $v_h(\y_h) > v_h(\x_h)$ for some $h\in N$. Then $\s{SW}(\y) > \s{SW}(\x)$. This gives $\sum_{j\in M'} \eta_j \ge \s{SW}(\y) > \s{SW}(\x) = \sum_{j\in M} \eta_j$, which is a contradiction. Hence $\x$ is PO. 
\end{proof}

\section{The Restricted Mixed Goods Setting}\label{sec:efxpo}

In this section, we investigate whether the fairness notions of \efx, \efxzero, \propmx~or \propmxzero~can be achieved in conjunction with the efficiency notion of Pareto-optimality (PO) in the mixed manna setting with restricted mixed goods. We first note that \propmxzero~and PO are not compatible by presenting an instance in Appendix~\ref{app:counter} for which no allocation is \propmxzero+PO, even when no item is a bad. This also shows that \efxzero+PO allocations needn't exist. Thus, we consider the existence of \propmx+PO and \efx+PO allocations.

\subsection{EFX+PO for Restricted Goods}
We begin by showing that we can obtain an EFX+PO allocation for the setting of \textit{pure goods with restricted valuations}. In such an instance $(N,M,V)$, for every $j\in M$, there exists a $v_j > 0$ s.t. $v_{ij}\in\{0,v_j\}$ for every $i\in N$.

\begin{algorithm}[h]
\caption{EFX+PO For Restricted Goods}\label{alg:envycycle}
\textbf{Input:} Restricted Goods Instance $(N,M,V)$\\
\textbf{Output:} Allocation $\x$
\begin{algorithmic}[1]
\State Order and relabel the goods so that $v_1 \ge v_2 \ge \dots \ge v_m > 0$ 
\State $\x \gets (\emptyset, \dots, \emptyset)$ \Comment{Initial empty allocation}
\While{$M \neq \emptyset$}
\State Pick $j\in M$
\State $N_j \gets \{i\in N: v_{ij} = v_j\}$ \Comment{Agents who value $j$ positively} 
\State Let $G_\x$ be the envy-graph defined by $\x$ \Comment{Def.~\ref{def:envy-graph}}
\State Let $G_j = G_\x[N_j]$ be the sub-graph of $G_\x$ induced by $N_j$
\State Let $i\in N_j$ be a \textit{source} in $G_j$ \Comment{Def.~\ref{def:envy-graph}} 
\State $\x_i \gets \x_i + j$ \Comment{Assign $j$ to $i$}
\State $M \gets M - j$
\EndWhile
\State \Return $\x$
\end{algorithmic}
\end{algorithm}

\begin{theorem}\label{thm:restrictedgoodsefx}
Given a fair division instance of pure goods with restricted valuations, an allocation that is EFX, PO and maximizes the utilitarian social welfare can be computed in polynomial-time.
\end{theorem}

\begin{proof}
We prove this theorem by showing that Algorithm~\ref{alg:envycycle}, which is based on the EnvyCycle procedure, computes an EFX+PO allocation for this setting. 

First note that Algorithm~\ref{alg:envycycle} always allocates a good $j$ to an agent $i$ who values it at the highest, i.e., $v_{ij} = v_j$. Hence by Lemma~\ref{lem:po}, the allocation is PO at every step of the algorithm.

We now show that the partial allocation $\x$ is EFX at every step of Algorithm~\ref{alg:envycycle}. Initially, $\x$ is an empty allocation (Line 2) and hence is trivially EFX. Inductively, let $\x$ be an EFX allocation of goods $1$ through $j-1$, for some $j\in [m]$, and let $G_\x$ be the envy-graph defined by $\x$. To maintain PO, we must allocate item $j$ to an agent $i$ s.t. $v_{ij} > 0$. Among the set of such agents $N_j$ for an item $j$ (Line 5), we allocate $j$ to an agent $i$ who is a source in the sub-graph $G_j$ of $G_\x$ induced by $N_j$. We argue below that such an agent always exists by showing that $G_\x$ is acyclic. Intuitively, $i$ is the right choice because giving the good $j$ to an agent $h\in N_j$ who is already envied by some agent $k$ can cause $k$ to envy $h$ even after the removal of $j$. 

More formally, we show that the allocation $\x'$ is EFX, where $\x'_i = \x_i + j$ and $\x'_h = \x_h$ for every $h\neq i$. Note that any violation of the EFX condition in $\x'$ must involve the agent $i$. Since $\x$ is EFX, we see that for any agent $h\neq i$, $v_i(\x'_i) = v_i(\x_i) + v_{ij} \ge v_i(\x_h - g) = v_i(\x'_h - g)$ for every $g\in \x_h$ with $v_{ig} > 0$. Thus $i$ does not EFX-envy any other agent in $\x'$. 

Consider $h\in N_j$. Since $i$ is a source in $G_j$, $h$ does not envy $i$. Hence $v_j(\x_h) \ge v_h(\x_i)$. Then for any $g\in \x'_i$ with $v_{hg}>0$, we have $v_h(\x'_i - g) = v_h(\x_i) + v_{hj} - v_{hg} \le v_h(\x_i) = v_h(\x'_h)$, since $\x$ is EFX and we ordered the goods so that $v_{hg}\ge v_{hj}$. Thus, $h$ does not EFX-envy $i$ in $\x'$.

Finally, consider $h\notin N_j$. Then $v_{hj} = 0$. Thus $v_h(\x'_h) = v_h(\x_h) \ge v_h(\x_i - g) = v_h(\x'_i - g)$ for any $g\in\x'_i$ s.t. $v_{hg} > 0$, since $\x$ is EFX. This shows that $h$ does not EFX-envy $i$ in $\x'$. 

In conclusion, $\x'$ is EFX. Inductively, this shows that the algorithm always maintains EFX. Further, since a good is given to an agent who values it positively, the allocation must be PO throughout the execution of the algorithm. We now show that the envy graph $G_\x$ corresponding to a partial allocation $\x$ is acyclic, ensuring the presence of a source agent (Line 8). Assume for the sake of contradiction, there is a cycle $C$ in $G_\x$. Then reallocating bundles along the cycle strictly improves the utility of every agent in the cycle, and does not change the utility of any agent not in the cycle. Hence, this is a Pareto-improvement over $\x$. This contradicts the fact that $\x$ is PO. Therefore, $G_\x$ cannot have cycles.
\end{proof}

\subsection{EFX+PO for Restricted Mixed Goods and Identical Bads}

We now use the result of the preceding section for computing EFX+PO allocations in the mixed setting comprising of \textit{restricted mixed goods and identical bads}. Recall that for such instances:
\begin{enumerate}
\item for every $j\in M^+$, there exists $v_j\in \Zp$ s.t. for every $i\in N$ with $v_{ij} > 0$, it holds that $v_{ij} = v_j$.
\item for every $j\in M^-$, there exists $v_j\in \Z^{-}$ s.t. for every $i\in N$, it holds that $v_{ij} = v_j$.
\end{enumerate}

\noindent We show that:
\begin{theorem}\label{thm:efxpo}
Given a fair division instance with restricted mixed goods and identical bads, an allocation that is EFX, \propmx, PO and maximizes the social welfare be computed in polynomial-time.
\end{theorem}

We prove Theorem~\ref{thm:efxpo} by showing that our Algorithm~\ref{alg:restrictedefxpo} computes an EFX+PO allocation for the given setting. 

\begin{algorithm}[h]
\caption{EFX+PO for Restricted Mixed Goods \& Id. Bads}\label{alg:restrictedefxpo}
\textbf{Input:} Instance $(N,M,V)$ with restricted mixed goods and identical bads\\
\textbf{Output:} Allocation $\x$
\begin{algorithmic}[1]
\State Partition $M$ into $M^+, M^0, M^-$ \Comment{See Sec.~\ref{sec:prelim}}
\Statex \textit{Phase 1: Allocating $M^+$ \dotfill}
\State Let $V'$ be given by
 \[v'_{ij} = 
    \begin{cases}
    v_{ij}, \text{ if } v_{ij} > 0 \\
    0, \text{ if } v_{ij} \le 0,
    \end{cases}\]
for every $i\in N$ and $j\in M^+$
\State $\x \gets \textsc{Algorithm~\ref{alg:envycycle}}(N,M^+,V')$ 
\Comment{Allocate $M^+$ by running Algorithm~\ref{alg:envycycle} for items in $M^+$ with modified values}
\Statex \textit{Phase 2: Allocating $M^0$ \dotfill}
\While{$M^0 \neq \emptyset$}
\State Pick $j\in M^0$ 
\State Let $i\in N$ be such that $v_{ij} = 0$
\State $\x_i \gets \x_i + j$ \Comment{Assign $j$ to $i$}
\State $M^0 \gets M^0 - j$
\EndWhile
\Statex \textit{Phase 3: Allocating $M^-$ \dotfill}
\State Order bads in $M^-$ according to $\prec$ s.t. $j\prec j'$ iff $-v_j \ge -v_{j'}$
\While{$M^- \neq \emptyset$}
\State Pick smallest $j\in M^-$ according to $\prec$ 
\State Let $G_\x$ be the envy-graph defined by $\x$ \Comment{Def.~\ref{def:envy-graph}}
\State Let $i$ be a \textit{sink} in $G_\x$ \Comment{Def.~\ref{def:envy-graph}}
\State $\x_i \gets \x_i + j$ \Comment{Assign $j$ to $i$}
\State $M^- \gets M^- - j$
\EndWhile
\State \Return $\x$
\end{algorithmic}
\end{algorithm}

We divide the execution of the algorithm into three phases: Phase 1 allocates items in $M^+$ (Lines 2-3), followed by Phase 2 for allocating items in $M^0$ (Lines 4-9), and finally Phase 3 for allocating items in $M^-$ (Lines 10-16). We show that at each iteration of the algorithm, the partial allocation of items allocated so far is always EFX+PO.

We first describe Phase 1. Given the instance $I = (N,M^+,V)$, we consider the instance $I' = (N,M^+,V')$ with values modified by changing the negative values to zero (see Line 2). Then, $I'$ is an instance of goods with restricted additive valuations. 

Using Algorithm~\ref{alg:envycycle}, we obtain an allocation $\x$ in Line 3 which is EFX+PO for the instance $I'$. We show that $\x$ is also EFX for $I$, thus showing that:
\begin{lemma}\label{lem:phase1}
The allocation at the end of Phase 1 is EFX+PO.
\end{lemma}
\begin{proof}
As discussed above, the allocation $\x$ at the end of Phase 1 is EFX+PO for the instance $I' = (N,M^+,V')$ with modified values. We now show that $\x$ is EFX for the original instance $I=(N,M^+,V)$ as well. To this end, observe that since $\x$ is PO for $I'$, no agent $i$ has an item $j$ which she values at 0. Otherwise transferring $j$ to an agent $h$ who values $j$ positively is a Pareto-improvement; such an agent $h$ exists because $j\in M^+$. Thus, $j\in \x_i$ implies $v'_{ij} > 0$. However, by construction this implies $v'_{ij} = v_{ij}$ (see Line 2). Thus, 
\begin{equation}\label{eqn:phase1.1}
v'_i(\x_i) = v_i(\x_i).
\end{equation}

Fix a pair of agents $i,h$. For any $j\notin \x_i$, $v_{ij} \le v'_{ij}$ by construction. Hence:
\begin{equation}\label{eqn:phase1.2}
v_i(\x_h- g) \le v'_i(\x_h-g),
\end{equation}
for any $g\in \x_h$ with $v_{ig} > 0$. Since $\x'$ is EFX~for $I'$, $v'_i(\x_i) \ge v'_i(\x_h-g)$. Together with \eqref{eqn:phase1.1} and \eqref{eqn:phase1.2}, we obtain:
\[v_i(\x_i) \ge v_i(\x_h-g),\]
for any $g\in \x_h$ with $v_{ig} > 0$. This shows that $\x$ is EFX~for $I$ at the end of Phase 1. Since each item $j$ is allocated to an agent $i$ who values it at the highest ($v_{ij}=v_j$), Lemma~\ref{lem:po} shows $\x$ is PO.
\end{proof}

Next, we describe Phase 2 which allocates items in $M^0$. To ensure that the allocation is PO, we must allocate each item $j\in M^0$ to an agent $i$ s.t. $v_{ij} = 0$. Intuitively, since the EFX condition requires that the envy between agents disappear after the removal of an item of \textit{positive} value, items of $M^0$ will not cause any new EFX-envy.

\begin{lemma}\label{lem:phase2}
The allocation at the end of Phase 2 is EFX+PO.
\end{lemma}
\begin{proof}
Suppose $\x$ is a partial EFX allocation prior to allocating an item $j\in M^0$. Suppose we allocate $j$ to an agent $i$ with $v_{ij} = 0$, to obtain an allocation $\x'$. Since $\x_h = \x'_h$ for every $h\neq i$, any possible violation of the EFX condition in $\x'$ must involve agent $i$. First, since $i$ does not get a bad, the utility of $i$ does not change, and neither does the utility of any other agent. Hence $i$ will not EFX-envy any other agent in $\x'$. Now since $\x$ is EFX, we have for any other $h\neq i$, $v_h(\x'_h) = v_h(\x_h) \ge v_h(\x_i - g) \ge v_h(\x'_i - g)$, for any $g\in \x'_i$ with $v_{hg}>0$. Hence $\x'$ is EFX.

Finally, Lemma~\ref{lem:po} once again shows that the allocation at the end of Phase 2 is PO, since any $j\in M^+ \cup M^0$ is allocated to an agent $i$ who values it at the highest-possible value.
\end{proof}

We finally describe Phase 3 which allocates items of $M^-$. Since the bads form an identical instance, i.e. a given bad $j$ has the same value $v_j < 0$ for all agents, we first sort the items in non-increasing order of disutility, i.e. we order the bads using $\prec$ where $j\prec j'$ iff $-v_{ij} \ge -v_{ij'}$. 

Let $\x$ be a partial EFX+PO allocation items in $M^+, M^0$ and some items of $M^-$. We must now decide how to allocate the next bad $j\in M^-$ in the order defined by $\prec$. For this, we consider the envy-graph $G=(N,E)$ (see Def.~\ref{def:envy-graph}). We argue that the way we allocate items ensures that $G$ will have at least one \textit{sink} agent, i.e., an agent who does not envy any agent. Then, we give $j$ to a sink $i$ (Line 14). Intuitively, this is the right choice, as $i$ is already well-off in terms of value as $i$ does not envy any other agent. Hence we should give the bad to $i$ instead of any other agent who has lesser value. We now formally show that:
\begin{lemma}\label{lem:phase3}
The allocation at the end of Phase 3 is EFX+PO.
\end{lemma}
\begin{proof}
Suppose $\x$ is a partial EFX allocation prior to allocating an item $j\in M^-$. Suppose we allocate $j$ to a sink agent $i$ in the envy-graph $G_\x$ to obtain an allocation $\x'$. We will show later that $G$ is acyclic, ensuring such a sink-agent exists. For sake of contradiction assume $\x'$ is not EFX. Since $\x_h = \x'_h$ for every $h\neq i$, any possible violation of the EFX condition must involve agent $i$. We claim that $i$ is EFX towards any other agent $h\neq i$ in $\x'$. To see this, note that for any bad $j'\in \x'_i$, $v_i(\x'_i - j') = v_i(\x_i) + v_{ij} - v_{ij'} \ge v_i(\x_i)$ since $j'<j$. Now since $i$ is a sink, $v_i(\x_i) \ge v_i(\x_h)$, thus implying that $v_i(\x'_i - j') \ge v_i(\x_h)$ for every bad $j'\in \x'_i$ and any agent $h\neq i$, showing that $i$ does not EFX-envy any other agent. 

Now suppose some agent $h\neq i$ EFX-envies $i$ in $\x'$. Since $\x$ was EFX, either (1) $v_h(\x_h) \ge v_h(\x_i- g)$ for every $g\in\x_i$ with $v_{hg} > 0$, or (2) $v_h(\x_h - j') \ge v_h(\x_i)$ for every $j'\in \x_h$ with $v_{hj'} < 0$. If Condition (1) holds, then we also obtain $v_h(\x'_h) = v_h(\x_h) \ge v_h(\x'_i- g)$ for every $g\in\x_i$ with $v_{hg} > 0$, since $\x'_h = \x_h$ and $\x'_i = \x_i + j$ but $v_{hj} \le 0$. Thus $h$ continues to be \efx~towards $i$ in $\x'$. If Condition (2) holds, then, using $v_{hj}\le 0$ we obtain $v_h(\x'_h - j') = v_h(\x_h - j') \ge v_h(\x_i) \ge v_h(\x'_i)$ for every $j'\in \x'_h$ with $v_{hj'} < 0$, once again showing that $h$ does not EFX-envy $i$ in $\x'$.

Thus $\x'$ is EFX. Further, note that any item is allocated to an agent who values it at the highest. By Lemma~\ref{lem:po}, the allocation remains PO throughout the execution of the algorithm. As argued before, if the envy-graph $G_\x$ corresponding to some partial allocation $\x$ has a cycle, then reallocating along the cycle gives a Pareto-improvement. Since $\x$ is always PO, this means that $G_\x$ is always acyclic, ensuring the presence of sink agents (Line 19). Thus, the allocation at the end of Phase 3 is EFX+PO.
\end{proof}

In conclusion, Algorithm~\ref{alg:restrictedefxpo} returns and EFX+PO allocation for the restricted mixed goods and identical bads setting. Note that any item is allocated to an agent who values it at the highest, hence the allocation also maximizes the social welfare by Lemma~\ref{lem:po}. Finally, Lemma~\ref{lem:efx-propm} implies that any EFX allocation is \propmx. This proves Theorem~\ref{thm:efxpo}.

\subsection{\efxzero+PO for Binary Mixed Goods and Identical Bads}\label{sec:binary}

Next, we show that for the special case of instances with binary mixed goods and identical bads, an \efxzero+PO allocation can be computed in polynomial time. 

\begin{theorem}\label{thm:efxzeropo}
Given a fair division instance with binary mixed goods and identical bads, an allocation that is \efxzero, \propmxzero, PO and maximizes the social welfare be computed in polynomial-time.
\end{theorem}

Recall that in this setting, there exists an $a>0$ such that for all $j\in M^+$, if $v_{ij} > 0$, then $v_{ij} = a$. In the restricted mixed goods setting, an \efxzero+PO allocation need not even exist (Appendix~\ref{app:counter}). Further, even for the binary mixed goods case, Algorithm~\ref{alg:restrictedefxpo} cannot provide \efxzero~guarantee because it may assign a mixed good $j$ to an agent $i$, and there could be another agent $h$ that envies $i$ but $v_{hj} = 0$, thus causing $h$ to \efxzero-envy $i$.

\begin{algorithm}[t]
\caption{\efxzero+PO for Binary Mixed Goods \& Id. Bads}\label{alg:binaryefxpo}
\textbf{Input:} Instance $(N,M,V)$ with binary mixed goods and identical bads\\
\textbf{Output:} Allocation $\x$
\begin{algorithmic}[1]
\State Partition $M$ into $M^+, M^0, M^-$ \Comment{See Sec.\ref{sec:prelim}}
\Statex \textit{Phase 1: Allocating $M^+$ \dotfill}
\State Let $V'$ be given by
\[v'_{ij} = 
    \begin{cases}
    v_{ij}, \text{ if } v_{ij} > 0 \\
    0, \text{ if } v_{ij} \le 0,
    \end{cases}\]
for every $i\in N$ and $j\in M^+$
\State $\x \gets \textsc{BinaryGoods}(N,M^+,V')$ \Comment{Allocate $M^+$ via \textsc{BinaryGoods} for items in $M^+$ with modified values}
\Statex \textit{Phase 2: Allocating $M^0$ \dotfill}
\While{$M^0 \neq \emptyset$}
\State Pick $j\in M^0$ 
\State Let $N_j \coloneqq \set{i \in N \midd v_{ij} = 0}$ 
\State Let $G_\x$ be the envy-graph defined by $\x$ \Comment{Def.~\ref{def:envy-graph}}
\State Let $G_j = G_\x[N_j]$ be the sub-graph of $G_\x$ induced by $N_j$
\State Let $i\in N_j$ be a \textit{source} in $G_j$ \Comment{Def.~\ref{def:envy-graph}} 
\State $\x_i \gets \x_i + j$ \Comment{Assign $j$ to $i$}
        $M^0 \gets M^0 - j$
\EndWhile
\Statex \textit{Phase 3: Allocating $M^-$ \dotfill}
\State Order bads in $M^-$ according to $\prec$ s.t. $j\prec j'$ iff $-v_j \ge -v_{j'}$
\While{$M^- \neq \emptyset$}
\State Pick smallest $j\in M^-$ according to $\prec$ 
\State Let $G_\x$ be the envy-graph defined by $\x$ \Comment{Def.~\ref{def:envy-graph}}
\State Let $i$ be a sink in $G$ \Comment{Def.~\ref{def:envy-graph}}
\State $\x_i \gets \x_i + j$ \Comment{Assign $j$ to $i$}
$M^- \gets M^- - j$
\EndWhile
\State \Return $\x$
\end{algorithmic}
\end{algorithm}

We circumvent this issue by using a different approach than Algorithm~\ref{alg:restrictedefxpo} for allocating $M^+$ in Phase $1$ and $M^0$ in Phase $2$. Specifically, we modify the values of items in $M^+$ as before to convert it into a binary goods instance $I'$, and then use algorithm \textsc{Alg-Binary} of \cite{amanatidis2020mnwefx}. This algorithm is adapted from \cite{darmann2014binary, barman2018binarynsw}, and computes a Nash-welfare maximizing allocation of $I'$ which is also \efxzero+PO (\cite{amanatidis2020mnwefx}, Theorem 3.1). Afterwards, we adapt Phase $2$ of Algorithm~\ref{alg:restrictedefxpo}, paying extra care to which agents we allocate the items in $M^0$. Finally, Phase $3$ of
Algorithm~\ref{alg:restrictedefxpo} remains the same. 

We prove Theorem~\ref{thm:efxzeropo} by showing that Algorithm~\ref{alg:binaryefxpo} computes an \efxzero+PO for the given setting. We first describe Phase 1 of Algorithm~\ref{alg:binaryefxpo}. Given the instance $I = (N,M^+,V)$, we consider the instance $I' = (N,M^+,V')$ with modified values,
as in the restricted setting. Notice that $I'$ is an instance with binary values. For such instances, an allocation 
which maximizes the Nash welfare can be computed in polynomial-time~\cite{darmann2014binary,barman2018binarynsw}, and further a
Nash welfare maximizing allocation which is also \efxzero+PO can be computed in polynomial-time~\cite{amanatidis2020mnwefx} via the procedure \textsc{Alg-Binary} of \cite{amanatidis2020mnwefx} (which we rename here to \textsc{BinaryGoods} to 
avoid confusion). Using this, we obtain an
allocation $\x$ in Line 3 which is \efxzero+PO for the instance $I'$. Like in Lemma~\ref{lem:phase1}, we show that $\x$ is also \efxzero~for $I$.

\begin{lemma}\label{lem:phase1-bin}
The allocation at the end of Phase 1 is \efxzero+PO.
\end{lemma}
\begin{proof}
As discussed above, we know that the allocation $\x$ at the end of Phase 1 is \efxzero+PO for the 
instance $I' = (N,M^+,V')$ with modified values. We now show that $\x$ is \efxzero~for the original instance $I=(N,M^+,V)$ as well. 
To this end, observe that, since $\x$ is PO for $I'$, no agent $i$ has an item $j$ which she values at 0. Otherwise transferring $j$ 
to an agent $h$ who values $j$ positively is a Pareto-improvement; such an agent $h$ exists because $j\in M^+$. Thus, $j\in \x_i$ 
implies $v'_{ij} > 0$. However, by construction this implies $v_{ij} > 0$ (see Line 2). Hence, for any agent $i$ and item $j$, if 
$j\in x_i$, then $v_{ij} > 0$, i.e., every agent gets items that are goods for them. Thus,
\begin{equation}\label{eqn:phase1.1-bin}
v'_i(\x_i) = v_i(\x_i).
\end{equation}
Fix a pair of agents $i, h \in N$. For any $j\notin \x_i$, $v_{ij} \le v'_{ij}$ by construction. Hence:
\begin{equation}\label{eqn:phase1.2-bin}
v_i(\x_h - g) \le v'_i(\x_h - g),
\end{equation}
for any $g \in \x_h$. Since $\x'$ is \efxzero~for $I'$, $v'_i(\x_i) \ge v'_i(\x_h - g)$. Together with 
\eqref{eqn:phase1.1-bin} and \eqref{eqn:phase1.2-bin}, we obtain:
\[
v_i(\x_i) \ge v_i(\x_h - g),
\]
for any $g\in \x_h$. This shows that $\x$ is \efxzero~for $I$ at the end of Phase 1. Further each item $j$ is allocated to some agent
$i$ who values it at the highest-possible value, thus showing that $\x$ is PO.
\end{proof}

Next, we describe Phase 2 which allocates items in $M^0$. The only difference between Algorithms~\ref{alg:restrictedefxpo} and~\ref{alg:binaryefxpo} in this Phase is that, while in the former we allocated every item $j \in M^0$ to an arbitrary agent which
has value $0$ for $j$, in the latter, out of all agents $i$ which have $v_{ij} = 0$, we allocate $j$ to a source of the
induced subgraph of the envy-graph. This ensures that no other agent $h$ who envies $i$ has $v_{hj} = 0$, which ensures the \efxzero~condition is not violated for $h$. We prove:
\begin{lemma}\label{lem:phase2-bin}
The allocation at the end of Phase 2 is \efxzero+PO.
\end{lemma}
\begin{proof}
Consider the partial allocation $\x$ right before we allocate an item $j \in M^0$, and assume that it is \efxzero+PO. Suppose we
allocate $j$ to $i$. By the algorithm's description, we know that $i$ is a source in $G_j$. First, since $i$'s valuation
for her bundle does not change, and no other agent's bundle changes, $i$ continues to satisfy the \efxzero~condition. Next,
we can partition $N - i$ into two sets, the set $N^E_{-i}$ of agents who envy $i$ and the set $N^{NE}_{-i}$ of agents who do
not envy $i$. For every agent $h \in N^{NE}_{-i}$, they continue to not envy $i$ even after we allocate $j$, since $v_h(\x_i)$
can not increase, while $v_h(\x_h)$ remains the same. Therefore, all $h \in N^{NE}_{-i}$ continue to satisfy the \efxzero~condition. Finally, since $i$ is a source in $G_j$, we know that every agent $h \in N^{E}_{-i}$ has $v_{hj} < 0$, else they would have an edge to $i$ and $i$ would not be a source. Therefore, $j$ is not considered among the items $g \in \x_i$ for which $v_{hg} \geq 0$, and thus all $h \in N^{E}_{-i}$ continue to satisfy the \efxzero~condition. 
\end{proof}

Phase 3 of Algorithm~\ref{alg:binaryefxpo} is exactly the same as Phase 3 of Algorithm~\ref{alg:restrictedefxpo} and, since
allocating pure bads does not affect items for which agents have value $0$, the allocation is \efxzero+PO after Phase 3, due to
Lemma~\ref{lem:phase3}. 

In conclusion, for instances with binary mixed goods and identical bads, Algorithm~\ref{alg:binaryefxpo} computes an \efxzero+PO 
allocation in polynomial-time. Algorithm~\ref{alg:binaryefxpo} allocates an item to an agent who values it at the highest, hence the allocation also maximizes the social welfare by Lemma~\ref{lem:po}. Finally, Lemma~\ref{lem:efx-propm} implies that any \efxzero~allocation is \propmxzero.
This proves Theorem~\ref{thm:efxzeropo}.

\section{\propmxzero~for Separable Instances}\label{sec:separ}

In this section, we consider separable instances, in which all agents agree on the set of goods and bads.
In other words, we can partition $M$ into $M^{\geq 0}$ and $M^-$ such that $v_{ij} \geq 0$ for all $j \in M^{\geq 0}$
and $v_{ij} < 0$ for all $j \in M^{\geq 0}$, for every agent $i \in N$. We show that:

\begin{theorem}\label{thm:ido-separ}
Given a separable fair division instance $(N,M,V)$, a \propmxzero~allocation can be computed in polynomial-time.
\end{theorem}

We prove this theorem by presenting an algorithm which returns a \propmxzero~
allocation for this instance, combining ideas from previously known algorithms for instances with goods only 
\cite{BaklanovGGS} and pure bads instances \cite{LiLiWu}.

The idea behind the algorithm is to first obtain a \propmxzero~allocation with respect to $M^{\geq 0}$ and then start 
allocating $M^-$. We show how this allocation of bads can be performed while maintaining \propmxzero~in the
case of an IDO instance, in which all agents have the same ordinal preference for all bads in 
$M^-$. Recall that in an IDO instance, there exists an ordering of the bads in $M^-$ such that
for all agents $i \in N$, $v_{i1} \leq v_{i2} \leq \dots \leq v_{im^-}$, where $m^- = |M^-|$. Notice that in an
IDO instance, the cardinal preferences can be significantly different across different agents. Afterwards, in
Section~\ref{sec:gen-chores}, we present a reduction from \cite{LiLiWu}, from an instance which is IDO for 
bads to a general instance for bads, and show that it holds even in the mixed manna setting. This allows us to obtain
a \propmxzero~allocation for any separable instance.

\subsection{Separable Instances with IDO Bads}\label{sec:ido-chores}

Our algorithm first uses the algorithm of Baklanov et al. \cite{BaklanovGGS} for goods to obtain a
\propmzero~allocation with respect to $M^{\geq 0}$. Note that this allocation will also trivially be \propmxzero, since we
have only assigned items from $M^{\geq 0}$. Also, Baklanov et al. do not differentiate between \propm~and
\propmzero, and thus mention their result as giving a \propm~allocation, but in fact their allocation satisfies the stronger
guarantee of being \propmzero. Afterwards, we run the envy-cycle elimination algorithm of Li et al. \cite{LiLiWu} for pure bads, 
and show that it obtains an allocation that is \propmxzero~with respect to $M$.

For the remainder of this section, we use $\textsc{Goods}$ to refer to Algorithm $1$ of \cite{BaklanovGGS}, and we assume
that the bads in $M^-$ are ordered such that $v_{i1} \leq v_{i2} \leq \dots \leq v_{im^-}$ (recall that
$v_{ij} < 0$ for all $i \in N, j \in M^-$, i.e. the bads are ordered from most painful to least painful).
Furthermore, given any allocation $\x$ of a subset of items, we create the top-envy graph of $\x$, $G^*_\x$. Consider a cycle $C$ of $G^*_\x$. We call a \emph{reallocation according to $C$} a new allocation $\x^C$ where we reallocate the
bundles of $C$ backwards along the cycle.

\begin{algorithm}[h]
\caption{Separable Instances with IDO Bads}\label{alg:ido-chores}
\textbf{Input:} Separable instance $(N,M,V)$\\
\textbf{Output:} Allocation $\x$
\begin{algorithmic}[1]
\State $M^{\geq 0} \gets \{j \in M : \forall i \in N, \: v_{ij} \geq 0\}$
\State $M^- \gets M \setminus M^{\geq 0}$
\State $\x^{\geq 0} \gets \textsc{Goods}(M^{\geq 0})$ \Comment{\propmzero~Algorithm from \cite{BaklanovGGS}}
\State $\x \gets \x^{\geq 0}$
\For{$j \gets 1$ to $m^-$}
\While{$\nexists$ any sinks in $G^*_\x$}
\State Let $C$ be a cycle in $G^*_\x$
\State Reallocate bundles according to $C$ \Comment{See text}
\EndWhile
\State Let $i$ be a sink in $G^*_\x$
\State $\x_i \gets \x_i + j$
\EndFor
\State \Return $\x$
\end{algorithmic}
\end{algorithm}

Since $\textsc{Goods}$ runs in polynomial time, Algorithm~\ref{alg:ido-chores} runs in polynomial time as well. The following lemma about allocation $\x^{\ge 0}$ (Line 3) follows directly from \cite{BaklanovGGS}.

\begin{lemma}\label{lem:baklav}
For every agent $i \in N$,
\[
v_i(\x^{\geq 0}_i) + d^{\geq 0}_i(\x^{\geq 0}) \geq \frac{1}{n} \cdot v_i(M^{\geq 0}),
\]
where $d^{\geq 0}_i(\x^{\geq 0}) = \max_{i' \neq i} \min_{j \in \x^{\geq 0}_{i'}} {v_{ij}}$.
\end{lemma}

Next, we show that Algorithm~\ref{alg:ido-chores} maintains \propmxzero~with respect to the currently allocated set of items
after allocating a pure bad.

\begin{lemma}\label{lem:separ}
Let $M^-_{\leq j} = \set{1, 2, \dots, j} \subseteq M^-$ denote the set of bads allocated after the $j$-th step 
of the For-Loop in
Lines $5-12$. Then, after the $j$-th step, for all agents $i \in N$, either
\[
v_i(\x_i) + d_i(\x) \geq \frac{1}{n} \cdot v_i\prn{M^{\geq 0} \cup M^-_{\leq j}},
\]
where $d_i(\x) = \max_{i' \neq i} \min_\stack{j \in \x_{i'}}{v_{ij} \geq 0} {v_{ij}}$, or
$\forall c \in \x_i \setminus \x^{\geq 0}_i,$
\[
v_i(\x_i - c) \geq \frac{1}{n} \cdot v_i\prn{M^{\geq 0} \cup M^-_{\leq j}}.
\]
\end{lemma}
\begin{proof}
We prove the lemma via induction on $j$. Notice that before the For-Loop in Lines $5-12$ of Algorithm~\ref{alg:ido-chores},
the lemma holds, because of Lemma~\ref{lem:baklav}. Assume that the lemma holds right after the $(j-1)$-th step of the
For-Loop, and consider the allocation $\x$ immediately after the $j$-th step of the For-Loop. Let $i$ be the agent that received bad $j$, and let $\x^{(b)}_i$ denote $i$'s bundle right before step $j$.

First, consider an agent $i' \neq i$. Notice that if $\x_{i'}$ did not change during the $j$-th step of the For-Loop, then
the lemma holds for $i'$, by our induction hypothesis, as the left-hand side of both inequalities remained the same, while
the right-hand side decreased. Furthermore, the only way $\x_{i'}$ could have changed during the $j$-th step of the For-Loop
is via a top-envy cycle reallocation in $G^*_\x$. In this case, $i'$ received her bust bundle among all bundles of $\x$,
since $G^*_\x$ is the top-envy graph, which implies that $i'$ does not envy any other agent after the $j$-th step,
thus $v_{i'}(\x_h) \geq \frac{1}{n} \cdot v_{i'}\prn{M^{\geq 0} \cup M^-_{\leq j}}$ and $i'$ satisfies the lemma.

Next, consider agent $i$. Since $i$ was a sink in the top-envy graph $G^*_\x$ prior to step $j$ of the For-Loop, we know
that $i$ did not envy any other agent, and thus
$v_i(\x_i) \geq \frac{1}{n} \cdot v_i\prn{M^{\geq 0} \cup M^-_{\leq j-1}}$. This implies that, for every bad
$c \in \prn{\x^{(b)}_i \setminus \x^{\geq 0}_i} + j$, we have
\begin{align*}
v_i\prn{\x^{(b)}_i + j - c} &= v_i\prn{\x^{(b)}_i + j} - v_{ic} \\
&= v_i\prn{\x^{(b)}_i} + v_{ij} - v_{ic} \\
&\geq \frac{1}{n} \cdot v_i\prn{M^{\geq 0} \cup M^-_{\leq j-1}} + v_{ij}- v_{ic} \\
&= \frac{1}{n} \cdot v_i\prn{M^{\geq 0} \cup M^-_{\leq j}} -\frac{v_{ij}}{n} + v_{ij} - v_{ic} \\
&\geq \frac{1}{n} \cdot v_i\prn{M^{\geq 0} \cup M^-_{\leq j}},
\end{align*}
where the last inequality follows from the fact that $v_{ic} \leq v_{ij} \leq \prn{1-\frac{1}{n}} v_{ij}$, since the
instance is IDO for bads and $v_{ik} < 0$ for all $k \in M^-$. Therefore, $i$ satisfies the lemma after the $j$-th
step.
\end{proof}

\subsection{General Separable Instances}\label{sec:gen-chores}

Consider now a general separable instance. Our idea is to again use the $\textsc{Goods}$ algorithm of \cite{BaklanovGGS},
and then convert the subproblem of allocating $M^-$ into an IDO instance via a reduction. This reduction appears in 
Li et al. \cite{LiLiWu}, and is commonly used in designing approximation algorithms for MMS fair allocations 
\cite{BouveretLemaitre, BarmanMurthy, HuangLu}. We show here that it can be used in the mixed manna setting as well.

\begin{lemma}\label{lem:gen-separ}
If there exists a polynomial time algorithm that given any separable mixed manna instance with IDO bads computes 
a \propmx~(resp. \propmxzero) allocation, then there exists a polynomial time algorithm that given any separable mixed manna
instance (with general bads) computes a \propmx~(resp. \propmxzero) allocation.
\end{lemma}
\begin{proof}
Given any separable mixed manna instance $I=(N, M, V)$, we start by creating $I'=(N, M, V')$, a new separable mixed manna instance
with IDO bads. $V'$ is defined as follows. For all items that are not pure bads, the valuation stays the same for all agents, i.e.
$v'_{ij} = v_{ij}$ for all $j \notin M^-$. Without loss of generality, let $M^- = \set{1, 2, \dots, m^-}$. Next, let $\sigma_i(j) \in M$ be the
$j$-th lowest-valued pure bad under value function $v_i$ (i.e. $\sigma_i(1)$ is the ``worst'' bad for agent $i$), and let
$v'_{ij} = v_{i \sigma_i(j)}$. It follows easily that instance $I'$ has IDO bads, with $v'_{i1}\leq v'_{i2}\leq \ldots \leq v'_{im^-}$.
	
Afterwards, we run the algorithm for IDO instances on instance $I'$, and get a \propmx~(resp. \propmxzero) allocation $\x'$ with 
respect to $I'$. By definition, for every agent $i \in N$, we have either
\[
v'_i(\x'_i) + d_i(\x') \geq \frac{1}{n} \cdot v'_i(M),
\]
where $d_i(\x') = \max_{i' \neq i} \min_\stack{j \in \x'_{i'}}{v'_{ij} > 0} {v'_{ij}}$ (resp. $d_i(\x') = \\ \max_{i' \neq i} \min_\stack{j \in \x'_{i'}}{v'_{ij} \geq 0} {v'_{ij}}$), or
$\forall c \in \x'_i \cap M^-$,
\[
v'_i(\x'_i - c) \geq \frac{1}{n} \cdot v'_i(M).
\]
Next, we construct a \propmx~(resp. \propmxzero) allocation for instance $I$, using $\x'$. For every item $j \notin M^-$, assign it to the agent that gets it in $\x'$,
i.e. if $j \in \x'_i$, then assign $j$ to $\x_i$. For every pure bad $j \in M^-$, in sequential order from $j = m^-$ to $1$, let $i_j$ denote
the agent that gets bad $j$ in $\x'$. We let $i_j$ select her favourite unallocated bad, $c = \argmax_{1 \leq \ell \leq j} {v_{i\ell}}$ and
assign it to $\x_i$. We show that there is a bijection $f_i: \x_i \rightarrow \x_i'$ such that for any bad $e \in \x_i$, we have $v_{ie} \geq v'_{ij}$,
where $j = f_i(e)$. Recall that $v'_{ij} \geq v'_{i \ell}$ for all $\ell \leq j$, and that $j$ is the $j$-th lowest valued pure bad for agent $i$.
Notice that there must exist an unallocated bad $e'$ with value $v_{ie'} \geq v'_{ij}$; if $\sigma_i(j)$ is unallocated at the beginning of step
$j$, then $e' = \sigma_i(j)$. If $\sigma_i(j)$ was allocated, by the pigeonhole principle and because at most $j-1$ bads have been allocated
at the beginning of step $j$, there must exist an unallocated bad $\sigma_i(k)$, where $k < j$, and we know that
$v_{i\sigma_i(k)} \geq v_{i\sigma_i(j)}$. Thus, there must exist an unallocated bad $e'$ with value $v_{ie'} \geq v'_{ij}$. Since $e$ has maximum
value for $i$ among unallocated bads under value function $v_i$, we have $v_{ie} \geq v'_{ij}$. Note that this process also ensures
$|\x'_i| = |\x_i|$, $v'_i(M) = v_i(M)$ and $d_i(\x') = d_i(\x)$ for all agents $i \in N$.

For every agent $i \in N$, there are two cases.
\begin{itemize}
    \item Assume $i$ is \propmx~(resp. \propmxzero) satisfied in $\x'$ because 
    \[
    v'_i(\x'_i) + d_i(\x') \geq \frac{1}{n} \cdot v'_i(M).
    \]
    Notice that $\x'$ and $\x$ are identical with respect to items that are not pure bads, and also, $v_i(\x_i \cap M^-) \geq v'_i(\x'_i \cap M^-)$,
    which implies $v_i(\x_i) \geq v'_i(\x'_i)$. Since $v'_i(M) = v_i(M)$ and $d_i(\x') = d_i(\x)$, we get:
    \[
    v_i(\x_i) + d_i(\x) \geq \frac{1}{n} \cdot v_i(M),
    \]
    and $i$ is \propmx~(resp. \propmxzero) satisfied in $\x$.
    
    \item Assume now that $i$ is \propmx~(resp. \propmxzero) satisfied in $\x'$ because $\forall c \in \x'_i \cap M^-$,
    \[
    v'_i(\x'_i - c) \geq \frac{1}{n} \cdot v'_i(M).
    \]
    We know that $\forall c \in \x_i \cap M^-$, $v_{ic} \geq v'_{if_i(c)}$. Thus, for every $c \in \x_i \cap M^-$:
    \[
        v_i(\x_i - c) = \sum_\stack{\ell \in \x_i}{\ell \neq c} {v_{i \ell}} \geq \sum_\stack{\ell' \in \x'_i}{\ell' \neq f_i(c)} {v'_{i \ell'}} = v_i(\x' - f_i(c)) \geq \frac{1}{n} \cdot v_i(M),
    \]
    and $i$ is \propmx~(resp. \propmxzero) satisfied in $\x$.
\end{itemize}
We conclude that $\x$ is a \propmx~(resp. \propmxzero) allocation with respect to $I$.
\end{proof}

Lemmas~\ref{lem:separ} and \ref{lem:gen-separ} together prove Theorem~\ref{thm:ido-separ}.

\section{Revisiting Restricted Mixed Goods}\label{sec:efx-gen}

In this section, we revisit the restricted mixed good setting of Section~\ref{sec:efxpo} and use the reduction of Lemma~\ref{lem:gen-separ}. 

\subsection{\efx~for Restricted Mixed Goods with IDO Bads}

We show that we can obtain \efx~allocations for the restricted mixed goods setting with IDO bads. 
\begin{theorem}\label{thm:restrictedgoodsefx-ido}
Given a fair division instance of restricted mixed goods and IDO bads, an \efx~allocation can be computed in polynomial-time. Furthermore, for the case of binary mixed goods, a \efxzero~allocation can be computed in polynomial-time.
\end{theorem}

The algorithm is almost identical to Algorithm~\ref{alg:restrictedefxpo}, with the only Phase 3 changed to allocate the IDO bads in the order $1$ to $m^-$, where $v_{i1}\le v_{i2}\le \dots v_{i{m^-}}$, for every $i\in N$, like in Algorithm~\ref{alg:ido-chores}. We replace Phase 3 of
Algorithm~\ref{alg:restrictedefxpo} with the following Algorithm~\ref{alg:restrictedefx-ido}.

\begin{algorithm}[h]
\caption{EFX for Restricted Mixed Goods \& IDO Bads}\label{alg:restrictedefx-ido}
\textbf{Input:} Instance $(N,M,V)$ with binary mixed goods and identical bads\\
\textbf{Output:} Allocation $\x$
\begin{algorithmic}[1]
\Statex \textit{Phase 3: Allocating $M^-$ \dotfill}
\State Order bads in $M^-$ in IDO order $v_{i1} \leq v_{i2} \leq \hdots \leq v_{im^-}$
\While{$M^- \neq \emptyset$}
\State Pick smallest $j\in M^-$ according to $\prec$
\State Let $G^*_\x = (N, E)$ be the top-envy-graph defined by $\x$ \Comment*{Def.~\ref{def:envy-graph}}
\While{$\nexists$ any sinks in $G^*_\x$}
\State Let $C$ be a cycle in $G^*_\x$
\State Reallocate bundles according to $C$ \Comment{See text}
\EndWhile
\State Let $i$ be a sink in $G^*_\x$ \Comment{Def.~\ref{def:envy-graph}}
\State $\x_i \gets \x_i + j$ \Comment*{Assign $j$ to $i$}
\State $M^- \gets M^- - j$
\EndWhile
\State \Return $\x$
\end{algorithmic}
\end{algorithm}

\begin{proof}[Proof of Theorem~\ref{thm:restrictedgoodsefx-ido}]

Suppose $\x$ is a partial EFX allocation prior to allocating an item $j\in M^-$. Suppose we allocate $j$ to a sink agent $i$ in the top-envy-graph $G^*_\x$ to obtain an allocation $\x'$. We will show later that resolving cycles in $G^*_\x$ preserves EFX. For sake of contradiction assume $\x'$ is not EFX. Since $\x_h = \x'_h$ for every $h \neq i$, any possible violation of the EFX condition must involve agent $i$. We claim that $i$ is EFX towards any other agent $h\neq i$ in $\x'$. To see this, note that for any bad $j'\in \x'_i$, $v_i(\x'_i - j') = v_i(\x_i) + v_{ij} - v_{ij'} \ge v_i(\x_i)$ since $j'<j$. Now since $i$ is a sink, $v_i(\x_i) \ge v_i(\x_h)$, thus implying that $v_i(\x'_i - j') \ge v_i(\x_h)$ for every bad $j'\in \x'_i$ and any agent $h\neq i$, showing that $i$ does not EFX-envy any other agent. 

Now suppose some agent $h \neq i$ EFX-envies $i$ in $\x'$. Since $\x$ was EFX, either (1) $v_h(\x_h) \ge v_h(\x_i- g)$ for every $g\in\x_i$ with $v_{hg} > 0$, or (2) $v_h(\x_h - j') \ge v_h(\x_i)$ for every $j'\in \x_h$ with $v_{hj'} < 0$. If Condition (1) holds, then we also obtain $v_h(\x'_h) = v_h(\x_h) \ge v_h(\x'_i- g)$ for every $g\in\x_i$ with $v_{hg} > 0$, since $\x'_h = \x_h$ and $\x'_i = \x_i + j$ but $v_{hj} \le 0$. Thus $h$ continues to be \efx~towards $i$ in $\x'$. If Condition (2) holds, then, using $v_{hj}\le 0$ we obtain $v_h(\x'_h - j') = v_h(\x_h - j') \ge v_h(\x_i) \ge v_h(\x'_i)$ for every $j'\in \x'_h$ with $v_{hj'} < 0$, once again showing that $h$ does not EFX-envy $i$ in $\x'$.

Next we argue that resolving any cycles in $G^*_\x$ preserves EFX. Suppose we reallocate the bundles according to a top-envy cycle
$C$ in $G^*_\x$. For any agent $i$ who is not in the cycle, her bundle is not changed by the reallocation. Although other bundles are
reallocated, the items in each bundle do not change and thus $i$ does not EFX-envy any other agent. For any agent $i \in C$, she will 
obtain her best bundle in this partial allocation $\x$, since $G^*_\x$ is the top-envy graph, and hence the cycle-swapped allocation is
envy-free for agent $i$. Therefore, resolving any cycles in $G^*_\x$ preserves EFX. Thus $\x'$ is EFX.

The case of binary mixed goods follows by the same analysis and the fact that the partial allocation of Algorithm~\ref{alg:binaryefxpo},
after the end of Phase $2$ is \efxzero+PO.
\end{proof}

\subsection{\propmx~for Restricted Mixed Goods with General Bads}

Our final result shows that we can obtain \propmx~allocations for the restricted mixed goods setting with general bads. To achieve
this, we utilize our reduction in Lemma~\ref{lem:gen-separ} from general bads instances to IDO bads instances. Note that the
allocation is not guaranteed to be Pareto Optimal.

\begin{theorem}\label{thm:restrictedgoodsefx-gen}
Given a fair division instance of restricted mixed goods, a \propmx~allocation can be computed in polynomial-time. Furthermore, for the case of binary mixed goods, a \propmxzero~allocation can be computed in polynomial-time.
\end{theorem}

\begin{proof}
Given any fair division instance of mixed goods with restricted valuations $I = (N,M,V)$, we first turn it into a new fair division
instance of mixed goods with restricted valuations and IDO bads $I' = (N,M,V')$, via the reduction of Lemma~\ref{lem:gen-separ}
on the set of pure bads $M^-$. Notice that, as long as the algorithm for IDO bads that is executed on the $I'$ returns an
allocation $\x$ that is \propmx+PO, we know that any item $j \notin M^-$ has been allocated in $\x$ to an agent that has non-negative 
value for it. Thus, since the reduction of Lemma~\ref{lem:gen-separ} does not distinguish between $M^+$ and $M^0$, it still holds in 
this setting as well.

Next, notice that the partial allocation of Algorithm~\ref{alg:restrictedefxpo} after the end of Phase $2$ has allocated all
items $j \notin M^-$. Also, it is \efx+PO, and thus by Lemma~\ref{lem:efx-propm}, it is \propmx+PO. Therefore, one can use the
partial allocation of pure bads of Phase $3$ for $I'$ as a guide on how to allocate the items in $M^-$ and maintain \propmx,
as Lemma~\ref{lem:gen-separ} describes.

The case of binary mixed goods follows by the same analysis and the fact that the partial allocation of Algorithm~\ref{alg:binaryefxpo},
after the end of Phase $2$ is \propmxzero+PO.
\end{proof}

\section{Discussion}\label{sec:discussion}
In this paper, we study the fair and efficient allocation of an indivisible mixed manna. We measured efficiency through Pareto-optimality, and fairness through EFX and PropMX, where PropMX combines PropM for goods and PropX for bads. We obtained polynomial time algorithms to find allocations that satisfy a mix of these guarantees for several classes of instances, namely separable, restricted mixed goods, and binary mixed goods. 

Much like EFX+PO, settling the existence of PropM+PO allocations is a challenging open problem even for goods manna. One can also ask whether allocations satisfying more demanding fairness notions than PropM exist, both for goods and the more general mixed manna setting. Lastly, the existence of weighted-PropMX allocations for agents with unequal entitlements is an interesting direction for future work.




\bibliography{references}

\appendix
\section{Counterexamples}\label{app:counterexamples}

\subsection{Non-Existence of PropX for Goods}\label{app:counter-propx}

In this section, we present an example that shows PropX allocations need not exist for goods, even for the simpler setting of identical
valuations. The example is due to Aziz, Moulin and Sandomirskiy \cite{MoulinAnnual, AzizMoulinSandom}. We present it here as well for completeness.

Consider three agents $a$ and $b$, and five goods $g_1, g_2, g_3, g_4$ and $g_5$. The agents' valuations for the 
goods are the following:
\[
\begin{array}{|c|c|c|c|c|c|}
\hline
        & g_1   & g_2   & g_3   & g_4   & g_5   \\
\hline
     a  & 3     & 3     & 3     & 3     & 1     \\
     b  & 3     & 3     & 3     & 3     & 1     \\
     c  & 3     & 3     & 3     & 3     & 1     \\
\hline
\end{array}
\]
\vspace{1em}

Notice that the proportional share for every agent is $\frac{13}{3}$ and that in any balanced allocation one agent gets two items of
value $3$, one agent gets one item of value $3$ along with $g_5$, which has value $1$ and one agent gets only one item of value $3$. 
Without loss of generality, let $\x_a = \set{g_1, g_4}, \x_b = \set{g_2, g_5}$ and $\x_c = \set{g_3}$. Notice that $v_c(\x_c) = 3$, and
even if $c$ receives $g_5$, we have $v_c(\x_c) + v_{c5} = 4 < \frac{13}{3}$. Thus, there exists no PropX allocation for this instance.

\subsection{Non-Existence of \propmzero+PO}\label{app:counter}

In this section, we present an example that shows \propmzero+PO allocations need not always exist, even for the simpler setting of 
restricted goods and no bads. Consider two agents $a$ and $b$, and three goods $g_1, g_2$ and $g_3$. The agents' valuations for the 
goods are the following:
\[
\begin{array}{|c|c|c|c|}
\hline
        & g_1   & g_2   & g_3   \\
\hline
     a  & 1     & 0     & 2     \\
     b  & 0     & 1     & 2     \\
\hline
\end{array}
\]
\vspace{1em}

Notice that for any PO allocation $\x$, we have $g_1 \in \x_a$ and $g_2 \in \x_b$. Also notice that $v_a(M) = v_b(M) = \frac{3}{2}$. 
However, either $a$ or $b$ will not receive $g_3$. Assume without loss of generality that $g_3 \in \x_a$. We have $d_b(\x) = 0$, since
the maximin item for agent $b$ is $g_1$, and thus $v_b(\x_b) + d_b(\x) = 1 + 0 < 3/2$. Therefore, there exists no \propmzero+PO
allocation for this instance.


\end{document}